\documentclass[11pt, reqno]{amsart}
\usepackage[utf8]{inputenc}
\usepackage{parskip}

\usepackage[nospace,noadjust]{cite}
\usepackage{relsize}
\usepackage{cite}
\usepackage{color}
\usepackage[dvipsnames]{xcolor}

\usepackage{tikz-cd}
\usetikzlibrary{calc,fit,matrix,arrows,automata,positioning}
\usetikzlibrary{decorations.pathreplacing,angles,quotes}
\usepackage{colortbl}
\usepackage{hyperref, enumitem}
\hypersetup{
  colorlinks   = true, 
  urlcolor     = blue, 
  linkcolor    = Purple, 
  citecolor   = red 
}
\usepackage{amsmath,amsthm,amssymb}

\usepackage[margin=2.5cm]{geometry}

\usepackage{stmaryrd}

\usepackage{array}
\newcolumntype{x}[1]{>{\centering\arraybackslash\hspace{0pt}}p{#1}}

\theoremstyle{definition}
\newtheorem{theorem}{Theorem}[section]
\newtheorem{definition}[theorem]{{{Definition}}}
\newtheorem{example}[theorem]{{{Example}}}

\newtheorem{remark}[theorem]{{{Remark}}}

\newtheorem{corollary}[theorem]{{{Corollary}}}
\newtheorem{proposition}[theorem]{{{Proposition}}}
\newtheorem{lemma}[theorem]{{{Lemma}}}

\newcommand{\fq}{\mathbb{F}_{q}}
\newcommand{\fqm}{\mathbb{F}_{q^m}}
\newcommand{\nkdm}{[n,k,d]_{q^m/q}}
\newcommand{\numberset}{\mathbb}

\newcommand{\C}{\mathcal{C}}

\newcommand{\F}{\numberset{F}}

\newcommand{\wt}{\textnormal{wt}}

\newcommand{\Fm}{\F_{q^m}}

\newcommand{\Fqmk}{\mathbb{F}_{q^{m}}^{k}}
\newcommand{\Fqmn}{\mathbb{F}_{q^{m}}^{n}}
\newcommand{\Fqm}{\mathbb{F}_{q^{m}}}

\title{On the existence of linear rank-metric intersecting codes}
\usepackage[foot]{amsaddr}
\author{Martino Borello$^{1,2}$}
\author{Olga Polverino$^3$}
\author{Ferdinando Zullo$^3$}
\address{$^1$Universit\'e Paris 8, Laboratoire de G\'eom\'etrie, Analyse et Applications, LAGA, Universit\'e Sorbonne Paris Nord, CNRS, UMR 7539, France.}
\address{$^2$Inria Saclay, France.}
\address{$^3$Dipartimento di Matematica e Fisica, Universit\`a degli Studi della Campania ``Luigi Vanvitelli'', I--\,81100 Caserta, Italy}

\email{martino.borello@univ-paris8.fr, \{olga.polverino,ferdinando.zullo\}@unicampania.it}

\begin{document}

\maketitle

\begin{abstract}
Intersecting codes are a classical object in coding theory whose rank-metric analogue has recently been introduced. Although the definition formally parallels the Hamming-metric case, the structure and parameter constraints of rank-metric intersecting codes exhibit substantially different behavior. It was previously shown that a nondegenerate $[n,k,d]_{q^m/q}$ rank-metric intersecting code must satisfy $2k-1 \le n \le 2m-3$, and the tightness of the upper bound was left open.
Using the geometric interpretation of rank-metric codes via $q$-systems, we prove that the dual subspace associated with a rank-metric intersecting code must satisfy strong evasiveness properties. This connection allows us to derive new restrictions
on the parameters of such codes and to show that the bound $n=2m-3$ can be attained only when $k=3$ and $m\ge 6$.
More generally, we show that $n \leq 2m-\lfloor(k+4)/2\rfloor$.
Moreover, we obtain a geometric characterization of these extremal codes in terms of scattered $\F_q$-subspaces of $\F_{q^m}^3$.
As a consequence, the existence problem for
$[2m-3,3,d]_{q^m/q}$ rank-metric intersecting codes is reduced to the existence of scattered subspaces of dimension $m+3$.
Using known constructions of maximum scattered subspaces, we derive existence results when $m$ is even. Finally, we prove that $[6,3,3]_{q^5/q}$ rank-metric intersecting codes do not exist for any prime power $q$, thus resolving an open problem
posed in \cite{bartoli2025linear}.
\end{abstract}

\medskip



\medskip

\section*{Introduction}

Intersecting codes are a classical object in coding theory. 
In the Hamming metric, a linear code is said to be \textbf{intersecting} if the supports of any two nonzero codewords share at least one coordinate; \cite{katona1983minimal,miklos1984linear}. 
Since their introduction, intersecting codes have been studied extensively due to their rich combinatorial structure and their connections to several areas such as secret sharing schemes, oblivious transfer, separating systems, frameproof codes, and additive combinatorics; see e.g. \cite{massey1993minimal,brassard2002oblivious,blackburn2003frameproof,randriambololona20132,borello2025geometry,plagne2011application,scotti2024recent}. 
In the binary case, intersecting codes coincide with minimal codes, a class that has attracted significant attention over the past decades.

In parallel, the theory of rank-metric codes has emerged as a central topic in modern coding theory, driven both by applications, most notably in network coding \cite{silva2008rank} and post-quantum cryptography \cite{bartz2022rank}, and by deep connections with finite geometry; see \cite{bartz2022rank,polverino2020connections,sheekeysurvey}. In this paper, for rank-metric codes we will mean only $\fqm$-subspaces of $\fqm^n$ equipped with the rank distance over $\fq$.
Rank-metric codes admit a natural geometric interpretation via $q$-systems, as shown in \cite{Randrianarisoa2020ageometric} (see also \cite{sheekey2019scatterd,alfarano2021linearcutting}). This viewpoint has proven particularly fruitful in the study of extremal families such as MRD codes, minimal rank-metric codes, and related structures (see \cite{alfarano2021linearcutting,zini2021scattered}).

Motivated by these developments, \cite{bartoli2025linear} has introduced and investigated the notion of \textbf{intersecting codes in the rank metric}. 
In this setting, a rank-metric code is said to be rank-metric intersecting if the rank supports of any two nonzero codewords intersect nontrivially. 
While this definition formally parallels the Hamming-metric case, the resulting class of codes exhibits substantially different behavior. 
In particular, rank-metric intersecting codes do not coincide with minimal rank-metric codes, although the two families intersect, and their existence and structure depend in a delicate way on the parameters of the code.
A key insight in the rank-metric setting is the geometric characterization of intersecting codes in terms of $q$-systems proved in \cite{bartoli2025linear}. 
More precisely, a rank-metric code is rank-metric intersecting if and only if its associated $q$-system is not generated by the sum of its intersections with two hyperplanes. 
This geometric viewpoint leads to nontrivial constraints on the parameters of rank-metric intersecting codes, provides new constructions beyond those arising from distance considerations alone, and reveals the presence of parameter ranges where the existence problem remains open.
More precisely, the authors in \cite{bartoli2025linear} showed that for a rank-metric intersecting code in $\Fqmn$ of dimension $k$ we have that
\[ 2k-1\leq n\leq 2m-3. \]
The tightness of the lower bound was proved in the same paper, whereas the tightness of the upper bound was left as an open problem.

The main goal of this paper is to investigate the tightness of the upper bound
$n \le 2m-3$ for rank-metric intersecting codes and to determine when this bound
can actually be attained.
Our approach relies on the geometric interpretation of rank-metric codes via $q$-systems. By analyzing the structure of the $\F_q$-subspace associated with a rank-metric intersecting code and its dual, we obtain new structural constraints on such codes. In particular, we show that the dual subspace associated with a rank-metric intersecting code must satisfy strong evasiveness properties. This connection with evasive subspaces allows us to derive new bounds on the parameters of these codes.
As a consequence, we prove that the maximum possible length $n=2m-3$ can occur only in the case $k=3$ and $m\ge 6$.
More generally, we show that 
\[n \leq 2m-\lfloor(k+4)/2\rfloor,\]
see Theorem \ref{thm:k=3andm6} for a more precise statement.
Moreover, we obtain a geometric characterization of these extremal codes in terms of scattered subspaces.
More precisely, a nondegenerate $[2m-3,k,d]_{q^m/q}$ rank-metric intersecting code exists if and only if $k=3$, $m\ge 6$, and the dual system associated with the code is a scattered $\F_q$-subspace of $\F_{q^m}^3$ of dimension $m+3$; this forces the minimum distance to be $m-1$.
This characterization reduces the existence problem for such codes to the existence of scattered subspaces with prescribed dimension. Using known results on maximum scattered subspaces, we derive existence results when $m$ is even. Finally, we complete the analysis by showing that $[6,3,3]_{q^5/q}$ rank-metric intersecting codes do not exist for any prime power $q$, thus resolving an open problem posed in \cite{bartoli2025linear}.

The paper is organized as follows. In Section~\ref{sec:prel} we recall the necessary background on rank-metric codes, $q$-systems, and $\F_q$-subspaces of $\F_{q^m}^k$. In Section~\ref{sec:intersectingRMcodes} we review the notion of rank-metric intersecting codes and the known parameter bounds. Section~\ref{sec:characterizations} contains our main structural results, where we connect intersecting codes with evasive subspaces and derive new bounds. In Section~\ref{sec:existence} we study the existence problem for codes of length $2m-3$, relating it to scattered subspaces. Finally, in Section~\ref{sec:633codes} we analyze the case $[6,3,3]_{q^5/q}$ and prove its non-existence.

\bigskip

\section{Preliminaries}\label{sec:prel}

\subsection{Rank-metric codes and their geometry}

Rank-metric codes were introduced by Delsarte \cite{de78} in 1978 as subsets of matrices and they have been intensively investigated in recent years because of their applications; we refer to \cite{sheekeysurvey,polverino2020connections}.

In this section we will be interested in rank-metric codes in $\F_{q^m}^n$.

An essential concept is the rank support of an element of $\fqm^n$.
Let $\Gamma=(\gamma_1,\ldots,\gamma_m)$ be an ordered $\fq$-basis of $\F_{q^m}$. For any vector $x=(x_1, \ldots ,x_n) \in \F_{q^m}^n$ define the matrix $\Gamma(x)\in \F_{q}^{m \times n}$, where
$$x_{j} = \sum_{i=1}^m \Gamma (x)_{i,j}\gamma_i, \qquad \mbox{ for all } j \in \{1,\ldots,n\},$$
that is $\Gamma(x)$ is the matrix expansion of the vector $x$ with respect to the $\fq$-basis $\Gamma$ of $\F_{q^m}$ .
The \textbf{rank support} of $x$ is defined as the row span of $\Gamma(x)$:
$$\mathrm{supp}(x)=\mathrm{rowsp}(\Gamma(x)) \subseteq \fq^n.$$

As shown in \cite[Proposition 2.1]{alfarano2021linearcutting}, the support of a vector is independent of $\Gamma$, allowing us to talk about the support of a word without reference to $\Gamma$.

The \textbf{rank} (weight) $w(v)$ of a vector $v=(v_1,\ldots,v_n) \in \F_{q^m}^n$ is defined as $w(v)=\dim (\mathrm{supp}(v))$. 

A \textbf{(linear vector) rank-metric code} $\C $ is an $\F_{q^m}$-subspace of $\F_{q^m}^n$ endowed with the rank distance, where such a distance is defined as $d(x,y)=w(x-y)$, where $x, y \in \F_{q^m}^n$. 
For details and applications we refer to \cite{bartz2022rank,gorla2018codes}.

Let $\C \subseteq \F_{q^m}^n$ be a rank-metric code. We will write that $\C$ is an $[n,k,d]_{q^m/q}$ code (or $[n,k]_{q^m/q}$ code) if $k=\dim_{\F_{q^m}}(\C)$ and $d$ is its minimum distance, that is $d=\min\{d(x,y) \colon x, y \in \C, x \neq y  \}$.

By the classification of $\F_{q^m}$-linear isometries of $\F_{q^m}^n$ (see \cite{berger2003isometries}), we say that two rank-metric codes $\C,\C' \subseteq \F_{q^m}^n$ are \textbf{(linearly) equivalent} if and only if there exists a matrix $A \in \mathrm{GL}(n,q)$ such that
$\C'=\C A=\{vA : v \in \C\}$. 
The codes we will consider are \textbf{nondegenerate}, i.e. those for which the columns of any generator matrix of $\C$ are $\fq$-linearly independent.

The geometric counterpart of a rank-metric code is the so-called $q$-system. 
An $[n,k,d]_{q^m/q}$ \textbf{system} $U$ is an $\F_q$-subspace of $\F_{q^m}^k$ of dimension $n$, such that
$ \langle U \rangle_{\F_{q^m}}=\F_{q^m}^k$ and
$$ d=n-\max\left\{\dim_{\F_q}(U\cap H) \mid H \textnormal{ is an }\F_{q^m}\textnormal{-hyperplane of } \F_{q^m}^k\right\}.$$
Moreover, two $[n,k,d]_{q^m/q}$ systems $U$ and $U'$ are \textbf{equivalent} if there exists an $\F_{q^m}$-isomorphism $\varphi\in\mathrm{GL}(k,q^m)$ such that
$$ \varphi(U) = U'.$$

Thanks to the following result, we can construct a system starting from a nondegenerate rank-metric code and conversely.

\begin{theorem}[\cite{Randrianarisoa2020ageometric}] \label{th:connection}
Let $\C$ be a nondegenerate $[n,k,d]_{q^m/q}$ rank-metric code and let $G$ be a generator matrix.
Let $U \subseteq \F_{q^m}^k$ be the $\F_q$-span of the columns of $G$.
The rank weight of an element $x G \in \C$, with $x \in \F_{q^m}^k$ is
\begin{equation}\label{eq:relweight}
w(x G) = n - \dim_{\fq}(U \cap x^{\perp}),\end{equation}
where $x^{\perp}=\{y \in \F_{q^m}^k \colon \langle x, y\rangle=0\}.$ In particular,
\begin{equation} \label{eq:distancedesign}
d=n - \max\left\{ \dim_{\fq}(U \cap H)  \colon H\mbox{ is an } \F_{q^m}\mbox{-hyperplane of }\F_{q^m}^k  \right\}.
\end{equation}
\end{theorem}

The above result allows us to give a one-to-one correspondence between equivalence classes of nondegenerate $[n,k,d]_{q^m/q}$ codes and equivalence classes of $[n,k,d]_{q^m/q}$ systems, see \cite{Randrianarisoa2020ageometric} and also \cite{alfarano2021linearcutting}.
The system $U$ and the code $\C$ as in Theorem \ref{th:connection} are said to be \textbf{associated} and a system associated with a code $\C$ will usually be written as $U_{\C}$.

\subsection{Properties of $\fq$-subspaces in $\Fqmk$}

It is well-known that $\Fqmk$ can be seen as an $\Fqm$-vector space of dimension $k$ and as an $\fq$-vector space of dimension $mk$.
The theory of linear sets  is essentially based on the idea of studying $\fq$-subspaces of $\Fqmk$ and how they behave with respect to $\Fqm$-subspaces of $\Fqmk$.
For instance, given an $\fq$-subspace $U$ of $\Fqmk$ and $W$ an $\Fqm$-subspace of $\Fqmk$, we define the \textbf{weight} of $W$ in $U$ as
\[ \textrm{wt}_U(W)=\dim_{\fq}(W\cap U). \]
A classical problem is the following: given the subspace $U$, what are the possible weights of the subspaces of a fixed dimension with respect to $U$?
We refer to \cite{polverino2010linear,lavrauw2015field} for more details.
In this section, we will rephrase some definitions and properties of linear sets in vector terms in order to fulfill the scope of the paper.
We start with the definition of scattered subspaces, originally introduced by Blokhuis and Lavrauw in \cite{blokhuis2000scattered}.
An $\fq$-subspace of $\Fqmk$ is said to be \textbf{scattered} if
\[ \textrm{wt}_U(\langle v \rangle_{\Fqm})=\dim_{\fq}(U \cap \langle v \rangle_{\Fqm})\leq 1, \]
for any $v \in \Fqmk$.
The dimension of a scattered subspace is bounded as follows.

\begin{theorem}[\cite{blokhuis2000scattered}]\label{thm:scattbound}
    Let $U$ be a scattered $\fq$-subspace in $\Fqm^k$, then
    \[ \dim_{\fq}(U)\leq \frac{mk}2. \]
\end{theorem}

More recently, some attention has been paid to $\fq$-subspaces in $\Fm^k$ which are \emph{evasive subspaces}. 
This definition extends the notion of scattered subspaces, later extended to $h$-scattered subspaces in \cite{lunardon2017mrd,sheekeyVdV,csajbok2021generalising}; see also \cite{bartoli2021evasive,gruica2022generalised}.

\begin{definition}
    Let $\mathbb{V}(k,q^m)$ be an $\F_{q^m}$-vector space of dimension $k$ and let $U$ be an $\fq$-subspace of $\mathbb{V}(k,q^m)$ with $\dim_{\fq} (U)=n$. Let $h$ and $r$ be positive integers such that $h \in \{1,\ldots,k\}$ and $h \leq r \leq km$. We say that $U$ is an $(h,r)$\textbf{-evasive} subspace if  $$\dim_{\fq}(U\cap T)\leq r$$
    for each $h$-dimensional $\fqm$-subspace $T$. If $r=h$, an $(h,h)$-evasive subspace is called \textbf{$h$-scattered} subspace. Furthermore, if $h=1$, then a $1$-scattered subspace will be simply called a \textbf{scattered} subspace.
    \end{definition}
    
    By \cite[Theorem 2.3]{csajbok2021generalising}, if $U$ is an $h$-scattered  $\fq$-subspace of $\mathbb{V}(k,q^m)$ with $\dim_{\fq}(U) >k$, then  $$\dim_{\fq}(U) \leq \frac{km}{h+1}.$$ Subspaces that reach the equality in the previous bound are called {\bf maximum $h$-scattered subspaces}.

In particular, the following bound has been proved.

\begin{theorem}[Corollary 4.9 in \cite{bartoli2021evasive}]\label{thm:boundevasive}
    Let $U$ be an $\fq$-subspace of an $\F_{q^m}$-vector space of dimension $k$.
    If $k<m$ and is an $(h,r)$-evasive subspace in $\fqm^k$, then
    \[ \dim_{\fq}(U)\leq mk-\frac{mkh}{r+1}. \]
\end{theorem}

In particular, the dimension of $(1,r)$-evasive subspaces in $\fqm^k$ is bounded by
\[ \dim_{\fq}(U)\leq\frac{mkr}{r+1}. \]

When considering $2$-dimensional $\Fqm$-subspaces $W=P_1+P_2$ (where $P_1$ and $P_2$ are two distinct $1$-dimensional $\Fqm$-subspaces), from a simple linear algebra argument we have that, for any $\fq$-subspace $U \subseteq W$,
\[ \mathrm{wt}_U(P_1+P_2)\geq \mathrm{wt}_U(P_1)+\mathrm{wt}_U(P_2). \]
If the equality holds we say that $U$ is with \textbf{complementary weights}.
This means that the subspace $U$ is the direct sum of $P_1 \cap U$ and $P_2 \cap U$.
This definition was introduced in \cite{napolitano2022linear} and extended in \cite{adriaensen2023minimum}.

We conclude this section by describing the following duality operation, which preserves $\Fqm$-linearity when acting on the $\Fqm$-linear subspaces.

Let $\sigma \colon \Fqmk\times \Fqmk \rightarrow \mathbb{F}_{q^m}$ be a nondegenerate reflexive bilinear form over $\Fqmk$. Define
$\sigma' \colon \Fqmk \times \Fqmk \rightarrow \mathbb{F}_q$ by 
\[\sigma':(u,v)\mapsto \mathrm{Tr}_{q^m/q}(\sigma(u,v))=\sigma(u,v)+\sigma(u,v)^q+\ldots+\sigma(u,v)^{q^{m-1}} .\]
If we regard $\Fqmk$ as an $\F_q$-vector space, then $\sigma^\prime$ turns out to be a nondegenerate reflexive bilinear form on $\Fqmk$.
Let $\perp$ and $\perp'$ be the orthogonal complement maps defined by $\sigma$ and $\sigma'$ on the lattices of $\F_{q^m}$-linear and $\F_q$-linear subspaces, respectively.
The following properties hold (see \cite[Section~2]{polverino2010linear} for more details).

\begin{proposition}\label{prop:dualityproperties}
With the above notation,
\begin{itemize}
    \item[(i)] $\dim_{\F_{q^m}}(W)+\dim_{\F_{q^m}}(W^\perp)=k$, for every $\F_{q^m}$-subspace $W$ of $\Fqmk$.
    \item[(ii)] $\dim_{\F_{q}}(U)+\dim_{\F_{q}}(U^{\perp'})=mk$, for every $\F_{q}$-subspace $U$ of $\Fqmk$.
    \item[(iii)] $T_1\subseteq T_2$ implies  $T_1^{\perp'}\supseteq T_2^{\perp'}$, for every $\F_q$-subspaces $T_1,T_2$ of $\Fqmk$.
    \item[(iv)] $W^\perp=W^{\perp'}$, for every $\F_{q^m}$-subspace $W$ of $\Fqmk$.
    \item[(v)] Let $W$ and $U$ be an $\F_{q^m}$-subspace and an $\F_q$-subspace of $\Fqmk$ of dimension $s$ and $t$, respectively. Then
    \[\dim_{\F_q}(U^{\perp'}\cap W^{\perp'})-\dim_{\F_q}(U\cap W)=mk-t-sm, \]
    i.e. 
    \begin{equation}\label{eq:dualweight} \mathrm{wt}_{U^{\perp'}}(W^{\perp})-\mathrm{wt}_U(W)=mk-\dim_{\F_q}(U) -m \dim_{\F_{q^m}}(W). \end{equation}
\end{itemize}
\end{proposition}

\bigskip

\section{Rank-metric intersecting codes}\label{sec:intersectingRMcodes}

The notion of rank-metric intersecting code has been introduced in \cite{bartoli2025linear} as follows.

\begin{definition}
    A code $\C$ is \textbf{rank-metric intersecting} if for any $c,c' \in \C\setminus\{0\}$ we have
    \[\mathrm{supp}(c)\cap \mathrm{supp}(c')\ne \{0\}.\]
\end{definition}

In order to describe the geometric counterpart of these codes, we give the following definition.

\begin{definition}
    Let $U$ be an $[n,k]_{q^m/q}$ system. We say that $U$ is \textbf{2-spannable} if there exist two $\Fqm$-linear hyperplanes $H_1$ and $H_2$ in $\Fqmk$ such that
    \[ U=(H_1\cap U)+(H_2\cap U). \]
\end{definition}

The authors in \cite{bartoli2025linear} proved that the geometric counterpart of an rank-metric intersecting code is a system that is not $2$-spannable.

\begin{theorem}[Theorem 3.4 in \cite{bartoli2025linear}] \label{thm:not2spannable}
Let $\C$ be a nondegenerate $[n,k,d]_{q^m/q}$ code. The code $\C$ is rank-metric intersecting if and only if $U_\C$ is not $2$-spannable.
\end{theorem}

Combining the geometric point of view and some results from \cite{borello2025geometry}, the authors in \cite{bartoli2025linear} provided some bounds on the length of a rank-metric intersecting code.

\begin{theorem}[Corollary 2.4 and Theorem 4.7 in \cite{bartoli2025linear}]\label{thm:boundsn}
    Let $\C$ be a nondegenerate $[n,k,d]_{q^m/q}$ rank-metric intersecting code. Then
    \[ 2k-1 \leq n \leq 2m-3. \]
\end{theorem}

Also, they showed constructions of rank-metric intersecting codes for a large sets of the parameters.

\begin{theorem}[Theorem 4.4 in \cite{bartoli2025linear}]\label{thm:existence-old}
    There exists a nondegenerate $[n,k,d]_{q^m/q}$ rank-metric intersecting code if 
    \[ 2k-1 \leq n \leq 2m-2k+1. \]
\end{theorem}

Clearly, the natural question that the authors left open (see \cite[Question 4.11]{bartoli2025linear}) regards the existence of these codes if the length is between $2m-2k+2$ and $2m-3$, for $k\geq 3$.
We will answer to this question in Section \ref{sec:existence}, after we deal with some characterization results.

\begin{theorem}[Theorem 4.2 in \cite{bartoli2025linear}]\label{thm:4.2}
Let $\C$ be a nondegenerate $\nkdm$ rank-metric intersecting code. Then
$$k \leq d \leq m.$$
\end{theorem}

Let us remark that rank-metric intersecting codes are also studied in the very recent preprint \cite{conca2026intersecting}, where the authors relate them to the  concept of vertical
connectivity for $q$-matroids.

\bigskip

\section{Some characterization results}\label{sec:characterizations}

In this section we investigate structural constraints on rank-metric intersecting codes by analyzing the geometry of the $q$-system associated with the code and its dual. 

We first provide a geometric characterization of subspaces that are not $2$-spannable in terms of complementary weights of $\fq$-subspaces obtained by intersecting the dual of the $q$-system with $2$-dimensional $\F_{q^m}$-subspaces. This characterization will allow us to connect rank-metric intersecting codes with evasive subspaces. Using this perspective, we prove that the dual $q$-system associated with a rank-metric intersecting code must satisfy strong evasiveness properties. 
These constraints lead to new bounds on the parameters of such codes and, in particular, to restrictions on the possible length and dimension. 
As an application, we obtain a characterization of rank-metric intersecting codes of length $2m-3$ in terms of scattered subspaces.

We start with the following characterization.

\begin{proposition}\label{prop:nocompl}
Let $U$ be an $\fq$-subspace of $\Fqmk$ of dimension $n$.  The subspace $U$ is $2$-spannable if and only if there exists a $2$-dimensional $\Fqm$-subspace $\ell$ of $\Fqmk$ such that $\ell \cap U^\perp$ is a $q$-system with complementary weights.
\end{proposition}
\begin{proof}
Clearly, $U$ is not $2$-spannable if and only if for any pairs of distinct $\Fqm$-hyperplanes $H_1$ and $H_2$ of $\Fqmk$ we have
$$ (H_1\cap U)+(H_2\cap U) \subset U.$$
Hence, 
$$\wt_U(H_1)+\wt_U(H_2)-\wt_U(H_1\cap H_2)<n.$$
From Equation \eqref{eq:dualweight}, the above inequality reads as follows
$$ \wt_{U^{\perp'}}(H_1^\perp+H_2^\perp)>\wt_{U^{\perp'}}(H_1^\perp)+\wt_{U^{\perp'}}(H_2^\perp),$$
 i.e. the intersection of $U^{\perp'}$ with any $2$-dimensional $\Fqm$-subspace of $\Fqm^k$ is not complementary weight.
\end{proof}

Using the characterization of non-$2$-spannable subspaces obtained in
Proposition~\ref{prop:nocompl}, we derive a strong restriction on the
dual subspace associated with a rank-metric intersecting code.
In particular, it must have an evasive property with respect to the one-dimensional $\fqm$-subspaces of $\fqm^k$.

\begin{theorem}\label{prop:UCperpscattgen}
    Let $\C$ be a nondegenerate $[n,k,d]_{q^m/q}$ rank-metric intersecting code with $k\geq 3$.
    Then $U_{\C}^{\perp'}$ is a $(1,2m-n-2)$-evasive $\fq$-subspace in $\fqm^k$.
\end{theorem}
\begin{proof}
    Let $n=2m-r$, we want to show that $U_{\C}^{\perp'}$ is a $(1,r-2)$-evasive $\fq$-subspace.
    Observe that $U_{\C}^{\perp'}$ is an $\fq$-subspace of $\Fqmk$ with dimension $km-2m+r$.
    By contradiction suppose that $U_{\C}^{\perp'}$ is not a $(1,r-2)$-evasive subspace. Then there exists $P=\langle v \rangle_{\Fqm}$ with $v \in \Fqmk\setminus \{0\}$ such that $w_{U_{\C}^{\perp'}}(P)=h\geq r-1$.
    Every $2$-dimensional $\Fqm$-subspace of $\Fqmk$ through $P$ has weight at least $h+2$, otherwise the intersection between such a $2$-dimensional $\Fqm$-subspace and ${U_{\C}^{\perp'}}$ is a a $q$-system with complementary weights, a contradiction to Proposition \ref{prop:nocompl}.
    This means that the number of nonzero vectors in ${U_{\C}^{\perp'}}$ is at least the number of vectors in 
    \[ \bigcup_{W \in \mathcal{L}}((W \cap {U_{\C}^{\perp'}})\setminus P) \,\, \cup (P \cap U_{\C}^{\perp'}), \]
    where $\mathcal{L}$ is the set of all the $2$-dimensional $\Fqm$-subspaces of $\Fqmk$ through $P$.
    Therefore, we get
    \[ q^{(k-2)m+r}-1 \geq \frac{q^{m(k-1)}-1}{q^m-1}(q^{h+2}-q^h)+q^h-1,  \]
    from which we derive
    \[q^{m(k-2)+r-h}\geq 1 +\frac{q^{m(k-1)}-1}{q^m-1}(q^2-1),\]
    and so
    \begin{equation}\label{eq:condevasiver-2} 
    q^{m(k-1)}-q^{m(k-1)+r-h-2}+q^{m(k-2)+r-h-2}-q^{m(k-1)-2}+q^{m-2}-1 \leq 0. 
    \end{equation}
    Observing that
    \[ q^{m(k-1)}-q^{m(k-1)+r-h-2}+q^{m(k-2)+r-h-2}-q^{m(k-1)-2} > 0, \]
    since
    \[ q^{m(k-2)-2}(q^{m+2}-q^{m+r-h}+q^{r-h}-q^m)> 0, \]
    as $r-h < 2$, and $q^{m-2}-1\geq 0$, \eqref{eq:condevasiver-2} gives a contradiction.
\end{proof}

We now exploit the evasiveness property of the dual subspace to derive
constraints on the parameters of a rank-metric intersecting code.
First, we derive a bound on the minimum distance on a rank-metric intersecting code, extending Theorem \ref{thm:4.2}.

\begin{corollary}\label{cor:bounddistance}
    Let $\C$ be a nondegenerate $[n,k,d]_{q^m/q}$ rank-metric intersecting code with $k\geq 3$.
    Then
    \[ \max\{k,n-m+2\}\leq d \leq m. \]
\end{corollary}
\begin{proof}
   Let $n=2m-r$, then by Theorem \ref{prop:UCperpscattgen}, we know that $\mathrm{wt}_{U_{\C}^{\perp'}}(\langle v \rangle_{\fqm})=\dim_{\fq}(U_{\C}^{\perp'} \cap \langle v \rangle_{\fqm})\leq r-2$ for any $v \in \fqm^k\setminus\{0\}$.
    By applying Proposition \ref{prop:dualityproperties}, we derive that
    \[ \mathrm{wt}_{U_{\C}}(\langle v \rangle_{\fqm}^\perp)=\mathrm{wt}_{U_{\C}^{\perp'}}(\langle v \rangle_{\fqm})-mk+2m-r+m(k-1)=\mathrm{wt}_{U_{\C}^{\perp'}}(\langle v \rangle_{\fqm})+m-r\leq m+2. \]
    Therefore, by Theorem \ref{th:connection}, we have that
    \[d\geq 2m-r-(m+2)=n-m-2.\]
    The statement now follows by Theorem \ref{thm:4.2}.
\end{proof}

As a consequence of Theorem~\ref{prop:UCperpscattgen} and the bound
for evasive subspaces, we obtain the following bound on the parameters of $\C$.

\begin{theorem}\label{thm:k=3andm6}
    Let $\C$ be a nondegenerate $[n,k,d]_{q^m/q}$ rank-metric intersecting code with $k\geq 3$ and let $n=2m-r$. Then $k\leq 2(r-1)-\frac{r^2-r}m$ and in particular $n\leq 2m - \left\lfloor \frac{k+4}{2} \right\rfloor$.
\end{theorem}
\begin{proof}
    Let $U_{\C}$ be a system associated with $\C$.
    By Theorem \ref{prop:UCperpscattgen}, we know that $U_{\C}^{\perp'}$ is a $(1,r-2)$-evasive subspace of dimension $mk-2m+r$ in $\Fqmk$.
    Therefore, by Theorem \ref{thm:boundevasive} we have that
    \[ mk-2m+r \leq \frac{mk(r-2)}{r-1}, \]
    from which we derive that 
    \begin{equation} \label{eq: BoundOnK}
         k \leq 2r-\frac{r^2}m-2+\frac{r}m. \end{equation}
    Now, by  (\ref{eq: BoundOnK}) we get the following lower bound on $r$
    \[r\geq \frac{(2m+1)-\sqrt{\Delta}}{2}\]
    where $\Delta=(2m+1)^2-4m(k+2)$. Since $\Delta < (2m-k-1)^2$, we get
    \[n=2m-r< 2m-\frac{k+2}{2}\] and hence the assertion.
\end{proof}

Therefore, we have the following characterization for rank-metric intersecting code of length $2m-3$.

\begin{corollary}\label{cor:charact2m-3}
    Let $\C$ be a nondegenerate $[2m-3,k,d]_{q^m/q}$ with $k\geq 3$. The code $\C$ is rank-metric intersecting if and only if 
    \begin{itemize}
        \item $k=3$;
        \item $m\geq 6$;
        \item $U_{\C}^{\perp'}$ is a scattered subspace.
    \end{itemize}
    Moreover, $d=m-1$.
    \end{corollary}
    \begin{proof}
The necessary condition follows from Proposition \ref{prop:UCperpscattgen} and Theorem \ref{thm:k=3andm6}. 
For the converse, observe that  if ${U_{\C}^{\perp'}}$ is a scattered subspace of dimension $m+3$ in $\Fqm^3$, then every two-dimensional subspace $\ell \subseteq \Fqm^3$ intersects ${U_{\C}^{\perp'}}$ in a scattered subspace of dimension at least $3$. In particular, the intersection $\ell \cap {U_{\C}^{\perp'}}$ is not complementary weight. By Theorem \ref{prop:UCperpscattgen}, it follows that $U_{\C}$ is not $2$-spannable. The claim now follows from Theorem \ref{thm:not2spannable}.
For the last part, from Corollary \ref{cor:bounddistance} it follows that $d \in \{m-1,m\}$. The assertion then follows observing that if $d=m$ then $\C$ would be a one-weight code and so a rank-metric simplex code (see \cite{Randrianarisoa2020ageometric,alfarano2021linearcutting}), a contradiction with the length of $\C$.
\end{proof}

In the next section we will see some related existence results.

\bigskip

\section{Some existence results}\label{sec:existence}

In this section we investigate the existence of nondegenerate
$[2m-3,k,d]_{q^m/q}$ rank-metric intersecting codes with $k\geq 3$.
By Theorem~\ref{thm:k=3andm6}, such codes can exist only if $k=3$
and $m\geq 6$.

Therefore, by Corollary~\ref{cor:charact2m-3}, the existence problem of $[2m-3,3,d]_{q^m/q}$ rank-metric intersecting codes with $m\geq 6$ can be translated into a geometric problem, namely the existence of a
scattered $\F_q$-subspace of $\F_{q^m}^3$ of dimension $m+3$.
We now recall some known results on the existence of scattered
subspaces and derive the corresponding existence results for
rank-metric intersecting codes.

If $m$ is even, scattered $\F_q$-subspaces of $\F_{q^m}^3$ of dimension
$3m/2$ always exist, as proved in
\cite{bartoli2018maximum,csajbok2017maximum}.

\begin{theorem}[\cite{bartoli2018maximum,csajbok2017maximum}]
\label{thm:scattexist}
Suppose that $m$ is even. Then there exists a scattered $\F_q$-subspace
of $\F_{q^m}^3$ of dimension $3m/2$.
\end{theorem}

Observe that for $m\geq 6$ we have $m+3 \leq 3m/2$. Hence a scattered
subspace of dimension $3m/2$ contains scattered subspaces of dimension
$m+3$, and therefore the existence of the former guarantees the
existence of the latter.

If $m$ is odd, the situation is less understood. Only partial results
are known; we refer the reader to \cite{polverino2020connections} for a
survey of the currently known constructions and bounds.

As a direct consequence of Corollary~\ref{cor:charact2m-3} and
Theorem~\ref{thm:scattexist}, we obtain the following existence result.

\begin{corollary}
Suppose that $m\geq 6$ is even. Then there exists a nondegenerate
$[2m-3,3,d]_{q^m/q}$ rank-metric intersecting code.
\end{corollary}

This shows that the upper bound in Theorem~\ref{thm:boundsn}
is achieved for $k=3$ and every even $m\geq 6$.

\begin{example}\label{ex:scattmax}
The $[9,3,5]_{64/2}$ rank-metric code $\C$ generated by 
$$G=\left(\begin{matrix}
1&0&0&\alpha^{54}&\alpha^{18}&\alpha^{26}&\alpha^{32}&\alpha^{22}&\alpha^{19}\\
0&1&0&\alpha^{59}&\alpha^{12}&\alpha^{50}&\alpha^{49}&\alpha^{57}&\alpha^{5}\\
0&0&1&\alpha^{29}&\alpha^{56}&\alpha^{34}&\alpha^{61}&\alpha^{44}&\alpha^{54}
\end{matrix}\right),$$
where $\alpha^6=\alpha^4 + \alpha^3 + \alpha + 1$ is a primitive element of $\F_{64}$, is a rank-metric intersecting of length $2\cdot 6-3$. 
Indeed, $\C$ has as associated system a maximum scattered $U$ $\F_2$-subspace of $\F_{64}^3$.
Since $U^{\perp'}$ is a maximum scattered $\F_2$-subspace of $\F_{64}^3$, by Corollary \ref{cor:charact2m-3} we have that $\C$ is rank-metric intersecting.
\end{example}

\begin{remark}
The case where $m$ is odd remains largely open, since the existence of
scattered $\F_q$-subspaces of $\F_{q^m}^3$ with large dimension is not
fully understood. Any new construction of scattered subspaces of
dimension at least $m+3$ would immediately provide new examples of
rank-metric intersecting codes of length $2m-3$.
\end{remark}

\begin{remark}\label{rmk:punctur}
Let us remark that, for $m=6$, we know the existence of $[n,3]_{q^6/q}$ rank-metric intersecting codes for $5 \leq n \leq 7$ and for $n=9$. We also know that for $n<5$ and $n>9$, an $[n,3]_{q^6/q}$ rank-metric intersecting code cannot exist. However, the existence of an $[8,3]_{q^6/q}$ rank-metric intersecting code remains open. A natural question is whether such a code can be obtained by extending or puncturing (in the rank-metric sense) known constructions. However, this does not seem to be an easy task: an extensive search for $[8,3]_{64/2}$ codes obtained by puncturing the code in Example~\ref{ex:scattmax} fails to produce a rank-metric intersecting one.  
\end{remark}

\bigskip

\section{Non-existence of $[6,3,3]_{q^5/q}$ rank-metric intersecting codes}\label{sec:633codes}

As already observed in \cite{bartoli2025linear}, for $k=2$ or $k=3$ and $m \leq 4$, the set of values of $n$ for which a rank-metric intersecting code exists is known. The first open case occurs for $k=3$ and $m=5$.
In this case, by Theorem~\ref{thm:existence-old} and Theorem~\ref{thm:k=3andm6}, the only unknown value of $n$ for the existence of a rank-metric intersecting code is $n=6$. In \cite{bartoli2025linear}, it is proved that a $[6,3,3]_{32/2}$ rank-metric intersecting code does not exist, leaving the general case open (see Question~4.10). In the following, we show that a $[6,3,3]_{q^5/q}$ rank-metric intersecting code does not exist for any value of $q$.

To this aim we will need the projective description of a $q$-system; we will restrict to the case of the projective plane.

\begin{definition}
    Let $\Lambda=\mathrm{PG}(\fqm^3,\fqm)=\mathrm{PG}(2,q^m)$. 
Consider an $\fq$-subspace $U$ of $\fqm^3$ and define the set
\[ L_U=\{\langle {u} \rangle_{\fqm} \colon {u}\in U\setminus\{{0}\}\} \]
as an $\fq$-\textbf{linear set} of rank $\dim_{\fq}(U)$.
\end{definition}

Denote by $\tau_j$ the number of lines in $\mathrm{PG}(2,q^m)$ intersecting $L_U$ in $j$ points.
By a classical double counting argument, we obtain the following.

\begin{proposition}\label{prop:fundrel}
    Let $L_U$ be an $\fq$-linear set in $\mathrm{PG}(2,q^m)$, then 
    \[ \sum_{i=0}^{|L_U|} \tau_{i}=q^{2m}+q^m+1,\]
  \[\sum_{i=0}^{|L_U|} i\tau_{i}=|L_U|(q^m+1),\]
  \[\sum_{i=0}^{|L_U|} i(i-1)\tau_{i}=|L_U|(|L_U|-1).
    \]
\end{proposition}

If $L_U$ has rank $n$ then the weight of any point is bounded by $n$. Denote by $N_i(U)$ the number of points of $\Lambda$ having weight $i\in \{0,\ldots,n\}$ in $L_U$, the following relations hold:
\begin{equation}\label{eq:card}
    |L_U| \leq \frac{q^n-1}{q-1},
\end{equation}
\begin{equation}\label{eq:pesicard}
    |L_U| =N_1(U)+\ldots+N_n(U),
\end{equation}
\begin{equation}\label{eq:pesivett}
    N_1(U)+N_2(U)(q+1)+\ldots+N_n(U)(q^{n-1}+\ldots+q+1)=q^{n-1}+\ldots+q+1.
\end{equation}

To our aim, we need the following result, which is a consequence of the Main Theorem of \cite{de2022weight}.

\begin{lemma}\label{lem:q5}
    Let $L_U$ be an $\fq$-linear set of rank $5$ in $\mathrm{PG}(1,q^5)$ with points of weight at most two.
    Then
    \[
    N_2(U)\leq \begin{cases}
        q^2+1, & \text{if } q>2,\\
        6, & \text{if } q=2.
    \end{cases}
    \]
\end{lemma}

By means of a refined combinatorial argument on the linear set associated with a $[6,3,3]_{q^5/q}$ rank-metric intersecting code, we prove that such a code does not exist.

\begin{theorem}
   A $[6,3,3]_{q^5/q}$ rank-metric intersecting codes does not exist for any prime power $q$.  
\end{theorem}
\begin{proof} 
By Theorem \ref{thm:not2spannable}, there exists a $[6,3,3]_{q^5/q}$  rank-metric intersecting code if and only if there exists a $[6,3,3]_{q^5/q}$ system $U$ which is not 2-spannable.
Let $U$ be a $[6,3,3]_{q^5/q}$ system in $\F_{q^5}^3$, and 
 suppose  that $U$ is not $2$-spannable. Since $n=6$, $d=3$ and  $U$ is not 2-spannable, the following properties follow readily:\\
 1) Each two-dimensional $\F_{q^5}$-subspace (we will call it a line) of $\F_{q^5}^3$ has weight at least one and at most three in $U$. \\
2) For each one-dimensional $\F_{q^5}$-subspace (we will call it a point) of weight $0$ in $U$, at most one line  through that point has weight three in $U$.\\
3)  Either  $U$ is a scattered subspace, or $U$ has exactly one point of weight $2$, while  all remaining points have weight at most $1$ (otherwise we would have lines of weight larger than $3$).

Let $L_U$ be the $\F_q$-linear set defined by $U$ in the projective plane $\mathrm{PG}(2,q^5)$, and let $\tau_i$ be the number of lines of $\mathrm{PG}(2,q^5)$ intersecting $L_U$ in $i$ points.
By Proposition \ref{prop:fundrel}, for the integers $\tau_i$ the following identities hold

 \begin{equation} \label{eq:a}
  \sum_{i=0}^{|L_U|} \tau_{i}=q^{10}+q^5+1
\end{equation}

\begin{equation} \label{eq:b}
  \sum_{i=0}^{|L_U|} i\tau_{i}=|L_U|(q^5+1)
\end{equation}
\begin{equation} \label{eq:c}
  \sum_{i=0}^{|L_U|} i(i-1)\tau_{i}=|L_U|(|L_U|-1),
\end{equation} 

From 1), 2) and 3) together with \eqref{eq:pesicard} and \eqref{eq:pesivett}, we get that $\tau_i=0$ if $i\neq 1,q+1,q^2+1,q^2+q+1$. Also, if $U$ is a scattered subspace, then $\tau_{q^2+1}=0$ and $|L_U|=q^5+q^4+q^3+q^2+q+1$, while if $U$ has exactly one point of weight $2$, and  all remaining points have weight at most $1$, then  $\tau_{q^2+1}=q^3+q^2+q+1$ and $|L_U|=q^5+q^4+q^3+q^2+1$.
In this second case, from Equations (\ref{eq:a}), (\ref{eq:b}) and (\ref{eq:c})
we get 
\[\tau_{q+1}=
\frac{-q^{11}-q^{10}+q^9+3q^8+4q^7+3q^6+q^5+2q^4+2q^3+2q^2+q}{q^3+2q^2+q+1},\]

which gives a negative number for every prime power $q$, a contradiction.
Hence, $U$ is a scattered subspace and from Equations (\ref{eq:a}), (\ref{eq:b}) and (\ref{eq:c})
we get 
\begin{equation} \label{eq: number of lines}
    \tau_1=q^{10}-q^8-q^7-q^6+q^5+q^3, \quad \tau_{q+1}=q^8+q^7+q^6-q^4-2q^3-2q^2-q, \quad \tau_{q^2+q+1}=
q^4+q^3+2q^2+q+1.\end{equation}

Let $P$ be any point of $L_U$, so $P=\langle u\rangle_{\F_{q^5}}$ where $u \in U\setminus \{0\}$.
Consider the projection $L_{U_P}$ of $L_U$ from the point $P$ in the quotient space $\F_{q^5}^3/P$, i.e. $U_P=U+P$ in the quotient vector space $\F_{q^5}^3/P$. Then $L_{U_P}$ is an $\F_q$-linear set of rank $5$ in $\mathrm{PG}(1,q^5)=\mathrm{PG}(\F_{q^5}^3/P,\F_{q^5})$. By Lemma \ref{lem:q5}, the linear set $L_{U_P}$ admits at most $q^2+1$ points of weight $2$ when $q>2$ and at most $6$ when $q=2$.
This means that the number  $\tau_{q^2+q+1}^P$ of lines in $\mathrm{PG}(2,q^5)$ of weight $3$ in $U$ passing  through $P$ is at most $q^2+1$.

Now, let $\ell$ be a fixed line with weight $3$ in $U$. By $2)$, any other line with weight $3$ in $U$ intersects $\ell$ in a point of $L_U$.
Counting the number of lines of weight $3$ in $U$, for $q>2$ we obtain
\[\tau_{q^2+q+1}=1+\sum_{P\in \ell \cap L_U}(\tau_{q^2+q+1}^P-1) \leq q^2|\ell\cap L_U|+1=q^2 (q^2+q+1)+1=q^4+q^3+q^2+1,\]

which contradicts (\ref{eq: number of lines}).
A similar contradiction can be obtained when $q=2$.
\end{proof}

The above theorem answers Question 4.10 of \cite{bartoli2025linear}.

\bigskip

\section{Conclusions}

In this paper we investigated the existence problem for linear rank-metric intersecting codes.
Starting from the geometric characterization of intersecting codes in terms of $q$-systems
introduced in \cite{bartoli2025linear}, we analyzed the structure of the associated $\F_q$-subspaces and
their duals. We first showed that the dual subspace associated with a rank-metric intersecting code must
satisfy strong evasiveness properties. More precisely, if $\C$ is a nondegenerate
$[n,k,d]_{q^m/q}$ rank-metric intersecting code with $k \ge 3$, then the dual system
$U_{\C}^{\perp'}$ is a $(1,2m-n-2)$-evasive subspace of $\F_{q^m}^k$.
Using known bounds on evasive subspaces, this property allowed us to derive new restrictions
on the parameters of rank-metric intersecting codes and to refine the upper bound on the
length.
More precisely, 
\[n\leq 2m- \left\lfloor \frac{k+4}2\right\rfloor.\]
Moreover, we proved that the maximum possible length $n=2m-3$ can occur only in the
case $k=3$ and $m\ge 6$. We also obtained a geometric characterization of such codes:
a nondegenerate $[2m-3,k,d]_{q^m/q}$ rank-metric intersecting code exists if and only if
$k=3$, $m\ge 6$, and the dual system $U_{\C}^{\perp'}$ is a scattered subspace of
$\F_{q^m}^3$ of dimension $m+3$, and as a consequence $d=m-1$. This characterization allows us to translate the existence problem for these codes into
the existence problem for scattered $\F_q$-subspaces of $\F_{q^m}^3$ with prescribed
dimension. Using known constructions of maximum scattered subspaces, we showed that
rank-metric intersecting codes with parameters $[2m-3,3,d]_{q^m/q}$ exist for every even
$m\ge 6$, proving that the upper bound $n\le 2m-3$ is tight in this case.
The above argument can also be extended by considering maximum $h$-scattered subspaces. However, the corresponding rank-metric intersecting codes rely on the existence of such maximum $h$-scattered subspaces for certain parameter choices, which are currently unknown. Finally, we completed the analysis of the boundary cases by proving that
$[6,3,3]_{q^5/q}$ rank-metric intersecting codes do not exist for any prime power $q$.
This result provides a complete answer to an open problem posed in \cite{bartoli2025linear}.

Several questions remain open. In particular, when $m$ is odd the existence of scattered
subspaces of dimension at least $m+3$ in $\F_{q^m}^3$ is not fully understood, and therefore
the existence of rank-metric intersecting codes of length $2m-3$ in this case remains open. Moreover, as highlighted  in Remark \ref{rmk:punctur}, the existence of $[2m-4,3]_{q^m/q}$ rank-metric intersecting codes is open also in the case where $m$ is even.
A more refined study on the geometric point of view on the Delsarte dual of an intersecting rank-metric code may be fruitful; see e.g. \cite{marino2023evasive,borello2025delsarte}.

\bigskip

\section*{Acknowledgments}
This research was also supported by Bando Galileo 2024 – G24-216. The first author is partially supported by the ANR-21-CE39-0009 - BARRACUDA (French \emph{Agence Nationale de la Recherche}).
The last two authors were partially supported by the Italian National Group for Algebraic and Geometric Structures and their Applications (GNSAGA - INdAM).
This work has been partially written during the Opera 2026 conference in Bordeaux and so we acknowledge the support from the International Research Laboratory LYSM in partnership between CNRS and INdAM.

\bigskip

\bibliographystyle{alpha}
\bibliography{biblio.bib} 

\end{document}